\documentclass[twocolumn,showpacs,amsmath,amssymb,10pt,aps]{revtex4}
  
\usepackage[utf8]{inputenc}
\usepackage[T1]{fontenc}
\usepackage{amsmath}
\usepackage{amsfonts}
\usepackage{graphicx}
\usepackage{yfonts}
\usepackage{soul}
\usepackage{color}
\usepackage[normalem]{ulem}
\usepackage{amsthm}
\usepackage{bm}
\usepackage{bbm}
\usepackage{mathtools}
\usepackage{array}

\DeclareMathOperator{\der}{d}

\def\<{\langle}
\def\>{\rangle}
\usepackage{placeins}
\usepackage{enumitem}% http://ctan.org/pkg/enumitem

\newcommand{\Tr}{\mathrm{Tr}}
\def\oper{{\mathchoice{\rm 1\mskip-4mu l}{\rm 1\mskip-4mu l}
{\rm 1\mskip-4.5mu l}{\rm 1\mskip-5mu l}}}
\newtheorem{Theorem}{Theorem}

\newtheorem{Proposition}{Proposition}
\newtheorem{Example}{Example}
\usepackage{caption}
\begin{document}

\title{\bf %Environmental noise as a way to optimize
Engineering fidelity of the generalized Pauli channels via legitimate memory kernels}
%\title{\bf On the noisy memory kernel evolution that enhances the generalized Pauli channels fidelity}
\author{{Katarzyna Siudzi{\'n}ska and Dariusz Chru{\'s}ci{\'n}ski }}
\affiliation{ Institute of Physics, Faculty of Physics, Astronomy and Informatics \\  Nicolaus Copernicus University,
Grudzi\k{a}dzka 5/7, 87--100 Toru{\'n}, Poland}

\begin{abstract}
We analyze the fidelity of the generalized Pauli channels governed by memory kernel master equations. It is shown that, by appropriate engineering of parameters of the corresponding memory kernel, the quantum evolution with non-local noise can have higher fidelity than the corresponding purely Markovian evolution governed by the Markovian semigroup. Similar engineering can substantially influence the evolution of quantum entanglement, entropy, and quantum coherence.
\end{abstract}

\maketitle

\section{Introduction}

Recently, much effort has been devoted to the analysis of open quantum systems \cite{open1,open2,open3}.
No realistic system is perfectly isolated due to the interaction with external environment, and therefore it has to be treated as an open system. Assuming that the interaction between the system and the environment is sufficiently weak, the well-known Born-Markov approximation can be applied to the evolution equation. This way, one derives the celebrated Markovian master equation
\begin{equation}
\dot{\rho}(t) = \mathcal{L}[\rho(t)] ,
\end{equation}
with the Gorini-Kossakowski-Sudarshan-Lindblad (GKSL) generator \cite{GKS,L}
\begin{equation}\label{GKSL}
\mathcal{L}[\rho]= -i[H,\rho] + \sum_{\alpha}\gamma_\alpha \left(V_\alpha\rho V_\alpha^\dagger -\frac 12\{V_\alpha^\dagger V_\alpha,\rho\}\right)
\end{equation}
and $\gamma_\alpha$ being (positive) decoherence/dissipation rates (for an intriguing history and importance of the GKSL master equation see a recent review \cite{CP17}). The noise operators $V_\alpha$ are responsible for decoherence and dissipation phenomena. In general, such environmental noise has a detrimental impact on a variety of quantum information processing tasks.  As a result, the quantum error correction has gained a considerable relevance and ultimately became a separate field of research \cite{lidar13}. There has also been a rapid development of passive schemes to protect quantum states from noise, such as the decoherence-free subspaces (see for example \cite{lidar98}). However, the way we perceive the role of the environmental noise has changed radically due to the seminal paper \cite{verstraete}, where it was shown that dissipation can be used to enhance quantum information processing. In particular, quantum information can be encoded in a set of steady states of a strongly dissipative system and manipulated coherently by using an effective dissipation-projected Hamiltonian \cite{zanardi14, marshall16}. 
It was shown that the memory effects caused by environmental noise can improve the channel fidelity \cite{Bogna1}, and also play a significant role in quantum thermodynamics, influencing for example our ability of extracting the work \cite{Bogna2}. Authors of \cite{zanardi16} showed that it is possible to perform the universal quantum computations that are robust to certain types of errors. The robustness of adiabatic quantum computation was considered in \cite{Childs}.

In recent years, considerable efforts have been made to describe the quantum evolution beyond the standard Markovian master equation. In order to do this, one has to take into account the memory effects caused by the nontrivial influence of the environment (see \cite{NM1,NM2,NM3,NM4} for recent reviews). One popular approach uses the so-called  Nakajima-Zwanzig master equation \cite{nakajima,zwanzig}
\begin{equation}\label{NZ}
\dot{\rho}(t)=\int_0^tK(t-\tau)[\rho(\tau)]\der\tau,
\end{equation}
where the memory kernel $K(t)$ encodes nontrivial memory effects.

The algebraic structure of legitimate memory kernels is known only for a limited number of cases \cite{filip15,kasia17}. Yet, sufficient conditions for the memory kernel to generate dynamical maps have been determined \cite{darch16}. Also, large classes of legitimate quantum evolution have been provided, such as the quantum semi-Markov evolution \cite{darch17} or the quantum stochastic dynamics corresponding to non-Markovian classical processes \cite{vacchini16}.

A decrease of the error accumulation was achieved for the dissipative Markovian processes \cite{zanardi17} and their generalizations \cite{lidar05,darch17}. In particular, it was shown that adding noises to the Markovian evolution slows down the rate at which the state of the system approaches the steady state \cite{zanardi17}. This conclusion was based on investigating the minimal channel fidelity of the dynamical maps provided by the Markovian semigroup generator (\ref{GKSL}) and its extension by the non-local memory kernel. As examples, the authors provided the multipartite Pauli channels and the generalized amplitude damping channel. This remarkable result shows that, instead of overcoming the environmental noise, one can actually benefit from it.

In this paper, we analyze how the channel fidelity of the generalized Pauli channels \cite{nathanson07,kasia16,kasia17} changes in time for the
evolution governed by eq. (\ref{NZ}) with
\begin{equation}\label{}
K(t) = \delta(t)\, \mathcal{L} + \mathbb{K}(t).
\end{equation}
In the above equation, $\mathcal{L}$ is a legitimate Markovian generator, and the non-local term $\mathbb{K}$ does not involve the Dirac delta function, therefore it is purely non-local. Recall that, using Uhlmann's transition probability formula \cite{uhlmann}, one defines the fidelity \cite{jozsa} between two mixed quantum states $\rho$, $\sigma$ by
\begin{equation}\label{statesfidelity}
F(\rho, \sigma) = \left[ \Tr \left(\sqrt{\sqrt{\rho}\sigma\sqrt{\rho}}\right)\right]^2.
\end{equation}
Now, for a given quantum channel $\Lambda$, the extremal values of the channel fidelity on pure input states $|\psi\rangle$ are defined via
\begin{equation}\label{channelfidelity}
\begin{split}
f_{\min}(\Lambda):=\min_PF(P,\Lambda[P])=\min_P \Tr(P\Lambda[P]),\\
f_{\max}(\Lambda):=\max_PF(P,\Lambda[P])=\max_P \Tr(P\Lambda[P]),
\end{split}
\end{equation}
where $P= |\psi\rangle \langle \psi|$. Due to concavity, $f_{\min}(\Lambda)$ is also the minimal channel fidelity on mixed quantum states $\rho$ \cite{wilde,Zycz}. Henceforth, we refer to $f_{\min}$ and $f_{\max}$ simply as the minimal and maximal channel fidelity, respectively.

The extremal channel fidelities allow us to measure how much, in the best and worst case scenario, a given quantum channel distorts the initial quantum state. Therefore, the larger the value of the channel fidelity, the better the channel can preserve the quantum information sent through it. Measuring the channel fidelity and engineering optimal quantum channels are two of the current challenges in quantum information theory \cite{nielsen}. Naturally, in quantum information processing, we would like to engineer such quantum channels that are as close as possible to the identity channel. Through a proper choice of the system parameters, we demonstrate a number of cases where the channel fidelity of $\Lambda(t)$ generated by $K(t)$ is better than that of the Markovian channel generated by $\mathcal{L}$ itself. This way, we prove that non-local memory effects can be used to decrease the error rate associated with the quantum channel. Interestingly, a similar technique allows one to engineer the evolution of quantum entropy, quantum entanglement, and quantum coherence.

\section{Fidelity of the generalized Pauli channels}

Consider a $d$-dimensional Hilbert space that admits the maximal number of $d+1$ mutually unbiased bases (MUBs). It is known that this is the case for $d=p^r$ with a prime $p$ and a natural $r$ \cite{Wootters,MAX}. Recall that the bases $\{\psi_0^{(\alpha)},\dots,\psi_{d-1}^{(\alpha)}\}$ for $\alpha=1,\dots,d+1$ are mutually unbiased if for any $\beta\neq\alpha$,
\begin{equation}
\<\psi_k^{(\alpha)},\psi_l^{(\alpha)}\>=\delta_{kl},\qquad
|\<\psi_k^{(\alpha)},\psi_l^{(\beta)}\>|^2 = \frac 1d.
\end{equation}
Introducing the rank-1 projectors onto the MUB vectors $P^{(\alpha)}_l = |\psi^{(\alpha)}_l\>\< \psi^{(\alpha)}_l|$ allows us to define $d+1$ unitary operators
\begin{equation}\label{U}
U_{\alpha} = \sum_{l=0}^{d-1} \omega^{l} P_l^{(\alpha)} \ ,
\end{equation}
where $\omega = e^{2\pi i/d}$. Now, let us use them to construct $d+1$ completely positive maps
\begin{equation}
\mathbb{U}_\alpha[\rho] = \sum_{k=1}^{d-1}  U_{\alpha}^k \rho  U_{\alpha}^{k \dagger}.
\end{equation}
Finally, the generalized Pauli channel is defined as follows \cite{nathanson07,kasia16},
\begin{equation}\label{GPC}
\Lambda=p_0\oper+\frac{1}{d-1}\sum_{\alpha=1}^{d+1}p_\alpha\mathbb{U}_\alpha,
\end{equation}
where $p_\alpha$ is the probability distribution, and $\oper$ denotes the identity map. For $d=2$, one reproduces the Pauli channel
\begin{equation}\label{PC}
\Lambda[\rho]=\sum_{\alpha=0}^{3}p_\alpha\sigma_\alpha\rho\sigma_\alpha
\end{equation}
with $\sigma_0=\mathbb{I}$ and $\sigma_k$ being the Pauli matrices. The eigenvalue equations for the generalized Pauli channel read $\Lambda[\mathbb{I}]=\mathbb{I}$ and
\begin{equation}\label{GPC_eigenvalue_eq}
\Lambda[U_\alpha^k]=\lambda_\alpha U_\alpha^k,\qquad k=1,\ldots,d-1,
\end{equation}
with the eigenvalues
\begin{equation}\label{GPC_eigenvalues}
\lambda_\alpha=p_0+\frac{d}{d-1}p_\alpha-\frac{1}{d-1}\sum_{\beta=1}^{d+1} p_\beta.
\end{equation}
On the other hand, one can express the probability distribution in terms of the eigenvalue functions,
\begin{equation}\label{c1}
p_0=\frac{1}{d^2}\left[1+(d-1)\sum_{\alpha=1}^{d+1}\lambda_\alpha\right],
\end{equation}
\begin{equation}\label{c2}
p_\alpha=\frac{d-1}{d^2}\left[1+d\lambda_\alpha-\sum_{\beta=1}^{d+1} \lambda_\beta\right].
\end{equation}
Observe that $\Lambda$ is completely positive and trace preserving if and only if it satisfies the generalized Fujiwara-Algoet conditions \cite{Fujiwara, nathanson07, Zyczkowski}
\begin{equation}\label{Fuji-d}
-\frac{1}{d-1}\leq\sum_{\beta=1}^{d+1}\lambda_\beta\leq 1+d\min_{\beta}\lambda_\beta.
\end{equation}

\begin{Theorem}\label{THM}
The minimal and maximal channel fidelities on pure input states for the generalized Pauli channel $\Lambda$ defined by eq. (\ref{GPC}) are given by
\begin{align}
&f_{\min}(\Lambda)=\frac{1}{d}\left[1+(d-1)\lambda_{\rm min} \right],\label{A1}\\
&f_{\max}(\Lambda)=\frac{1}{d}\left[1+(d-1)\lambda_{\rm max} \right],\label{A2}
\end{align}
where $\lambda_{\rm min} = \min_\alpha \lambda_\alpha$ and $\lambda_{\rm max} = \max_\alpha \lambda_\alpha$.
\end{Theorem}

\begin{proof}
Let us take an arbitrary rank-1 projector and write it as
\begin{equation}
P=\frac 1d \left(\mathbb{I}+\sum_{\alpha=1}^{d+1}\sum_{k=1}^{d-1}x_{\alpha k}U_\alpha^k\right).
\end{equation}
When acting on such $P$, the channel $\Lambda$ transforms it into
\begin{equation}
\Lambda[P]=\frac 1d \left[\mathbb{I}+\sum_{\alpha=1}^{d+1}\sum_{k=1}^{d-1}
\lambda_\alpha x_{\alpha k}U_\alpha^k\right].
\end{equation}
Therefore, the channel fidelity of the generalized Pauli channel acting on $P$ reads
\begin{equation}\label{sch}
F(P,\Lambda[P])=\Tr(P\Lambda[P])=\frac{1}{d}\left(1+\sum_{\alpha=1}^{d+1}\lambda_\alpha
\sum_{k=1}^{d-1}|x_{\alpha k}|^2\right).
\end{equation}
We know that if $P$ is a rank-1 projector, then
\begin{equation}
\Tr P^2=\frac{1}{d}\left(1+\sum_{\alpha=1}^{d+1}\sum_{k=1}^{d-1}
|x_{\alpha k}|^2\right)=1,
\end{equation}
and hence
\begin{equation}
\sum_{\alpha=1}^{d+1}\sum_{k=1}^{d-1}|x_{\alpha k}|^2=d-1.
\end{equation}
Therefore, the minimal value of $F(P,\Lambda[P])$ is attained when
\begin{equation*}
x_{\alpha k} = 0\qquad\mbox{for}\qquad\alpha \neq \alpha_m
\end{equation*}
with $\alpha_m$ corresponding to $\lambda_{\alpha_m}=\lambda_{\rm min}$. Similarly, the maximal value of $F(P,\Lambda[P])$ is attained when
\begin{equation*}
x_{\alpha k} = 0\qquad\mbox{for}\qquad\alpha \neq \alpha_M
\end{equation*}
with $\alpha_M$ corresponding to $\lambda_{\alpha_M} = \lambda_{\rm max}$.
\end{proof}

%\section{Local and non-local evolution}

\section{Markovian semigroup vs. general dynamical map}

Consider the generalized Pauli channel evolution of the density matrix governed by the Markovian semigroup $\Lambda^{\rm MS}(t) = e^{t \mathcal{L}}$. This evolution is given by $\rho\longmapsto\rho(t)=\Lambda^{\rm MS}(t)[\rho]$, where $\{\Lambda^{\rm MS}(t)|t\geq 0\}$ is the family of the generalized Pauli channels with the initial condition $\Lambda^{\rm MS}(0)=\oper$. Clearly, it satisfies the semigroup property
\begin{equation*}
\Lambda^{\rm MS}(t) \Lambda^{\rm MS}(s) =\Lambda^{\rm MS}(t+s).
\end{equation*}
The corresponding time-independent operator $\mathcal{L}$ is given by
\begin{equation}\label{GEN}
\mathcal{L}=\sum_{\alpha=1}^{d+1}\gamma_\alpha\mathcal{L}_\alpha,\qquad\mathcal{L}_\alpha=\frac 1d \left[\mathbb{U}_\alpha-(d-1)\oper\right].
\end{equation}
The eigenvalues of $\mathcal{L}$ read $\mu_0=0$ and
\begin{equation}
\mu_\alpha=\gamma_\alpha-\gamma_0,
\end{equation}
where $\gamma_0=\sum_{\alpha=1}^{d+1}\gamma_\alpha$, $\mathcal{L}[U^k_\alpha]=\mu_\alpha U^k$.
Hence, the eigenvalues $\lambda^{\rm MS}_\alpha(t)$ of the corresponding $\Lambda^{\rm MS}(t)$ are equal to
\begin{equation}\label{eigenvalue2}
\lambda^{\rm MS}_\alpha(t)=\exp[-(\gamma_0-\gamma_\alpha)t].
\end{equation}
The more general evolution, which includes the memory effects, is provided by the memory kernel master equation
\begin{equation}\label{memory_kernel_equation}
\dot{\Lambda}(t)=\int_0^tK(t-\tau)\Lambda(\tau)\der\tau ,
\end{equation}
with the memory kernel
\begin{equation}\label{K}
K(t)=\sum_{\alpha=1}^{d+1}k_\alpha(t)\mathcal{L}_\alpha.
\end{equation}
From the eigenvalue equations of the memory kernel,
\begin{equation}\label{kernel_eigenvalue_eq}
K(t)[U_\alpha^k]=\kappa_\alpha(t) U_\alpha^k, \qquad K(t)[\mathbb{I}]=0,
\end{equation}
we see that it shares its eigenvectors with $\Lambda(t)$. The corresponding eigenvalues are
\begin{equation}
\kappa_\alpha(t)=k_\alpha(t)-k_0(t)
\end{equation}
with $k_0(t)=\sum_{\beta=1}^{d+1}k_\beta(t)$. Therefore, eq. (\ref{memory_kernel_equation}) is equivalent to the following evolution equation for the eigenvalues of $K(t)$ and $\Lambda(t)$,
\begin{equation}
\dot{\lambda}_\alpha(t)=\int_0^t\kappa_\alpha(t-\tau) \lambda_\alpha(\tau)\der\tau,
\end{equation}
with $\lambda_\alpha(0)=1$. Using the Laplace transform method of solving differential equations, we find
\begin{equation}\label{lambda_s}
\widetilde{\lambda}_\alpha(s)=\frac{1}{s-\widetilde{\kappa}_\alpha(s)},
\end{equation}
where $\widetilde{f}(s)=\int_0^\infty f(t)e^{-st} \der t$ is the Laplace transform of $f(t)$. Now, let us introduce the following parametrization,
\begin{equation}
\lambda_\alpha(t) = 1 - \int_0^t \ell_\alpha(\tau) \der \tau .
\end{equation}
In \cite{kasia17}, the authors provided the necessary and sufficient conditions for the legitimate memory kernels that generate the generalized Pauli dynamical maps.

\begin{Theorem}\label{TH1}
The memory kernel $K(t)$ given in eq. (\ref{K}) gives rise to a legitimate generalized Pauli dynamical map $\Lambda(t)$ if and only if its eigenvalues $\kappa_\alpha(t)$ are equal to
\begin{equation}\label{th_1}
\widetilde{\kappa}_\alpha(s)=-\frac{s \widetilde{\ell}_\alpha(s)}{1-   \widetilde{\ell}_\alpha(s)},
\end{equation}
where the functions $\ell_\alpha(t)$ satisfy the conditions
\begin{eqnarray}\label{CON}
   \int_0^t \ell_\alpha(\tau) \der \tau &\geq & 0, \nonumber \\
  \sum_{\alpha=1}^{d+1}   \int_0^t \ell_\alpha(\tau) \der \tau &\leq & \frac{d^2}{d-1}, \\
   \sum_{\alpha=1}^{d+1}   \int_0^t \ell_\alpha(\tau) \der \tau & \geq &  d\, \int_0^t \ell_\beta(\tau) \der \tau \nonumber
\end{eqnarray}
for $\beta=1,\ldots,d+1$.
\end{Theorem}

\section{Examples}

%In this section, we analyze how the behaviour of the channel fidelity changes after we add the non-local effects to the Markovian evolution. Namely, we consider the memory kernels of the form
%\begin{equation}
%K(t)=\mathcal{L}\delta(t)+\mathbb{K}(t),
%\end{equation}
%where the Markovian semigroup generator $\mathcal{L}$ describes the time-local evolution, and $\mathbb{K}(t)$ is purely non-local (has no parts with the Dirac delta).

\subsection{Oscillations}

In \cite{zanardi17}, the authors analyzed the Pauli channels $\Lambda(t)$ whose evolution is governed by $K(t)=\mathcal{L}\delta(t)+\mathbb{K}(t)$ with
\begin{equation}\label{KL}
\mathcal{L}=\gamma\mathcal{L}_{\alpha_\ast},\qquad\mathbb{K}(t)=k(t)\mathcal{L}_{\alpha_\ast}
\end{equation}
for a fixed $\alpha_\ast\in\{1,2,3\}$, where $\mathcal{L}_{\alpha_\ast}$ is given by formula (\ref{GEN}) for $d=2$. An important property of this memory kernel is that both $\mathcal{L}$ and $\mathbb{K}(t)$ generate legitimate solutions. As the memory function $k(t)$, one considers \cite{zanardi17}
\begin{equation}
k(t)=\gamma B^2 e^{-t/T},
\end{equation}
where the constants $\gamma$, $B$, and $T$ are positive.

This example can be easily generalized to the generalized Pauli channels. One simply replaces the generator $\mathcal{L}_{\alpha_\ast}$ with a general $\mathcal{L}_{\alpha_\ast}$, $\alpha_\ast\in\{1,\dots,d+1\}$ defined in eq. (\ref{GEN}).
Now, observe that the memory kernel $K(t)$ is associated with $\ell_{\alpha_\ast}(t)=0$ and
\begin{equation}\label{ell}
\ell_\alpha(t)=\displaystyle\frac{\gamma}{\zeta}e^{-\frac{(1+\gamma T)t}{2T}}\Bigg[(1+2B^2T-\gamma T)
\sin\displaystyle\frac{\zeta t}{2T}
+\zeta\cos\displaystyle\frac{\zeta t}{2T}\Bigg]
\end{equation}
for every $\alpha\neq\alpha_\ast$ with (possibly complex)
\begin{equation}
\zeta:=\sqrt{-(1-\gamma T)^2+4\gamma B^2T^2}.
\end{equation}
For the corresponding dynamical map $\Lambda(t)$, the eigenvalues $\lambda_{\alpha_\ast}(t)=1$ and $\lambda_\alpha(t)\equiv\lambda(t)$ for $\alpha\neq{\alpha_\ast}$. Therefore, the maximal and minimal channel fidelities are equal to $f_{\max}[\Lambda(t)]=1$ and
\begin{equation}\label{minf}
f_{\min}[\Lambda(t)]=\frac 1d \left[1+(d-1)\lambda(t)\right].
\end{equation}
The eigenvalues
\begin{equation}\label{lambda1b}
\lambda(t)=\frac{2BT\sqrt{\gamma}}{\zeta}
e^{-\frac{(1+\gamma T)t}{2T}}\cos\left(\frac{\zeta t}{2T}+\arctan\frac{\gamma T-1}{\zeta}\right)
\end{equation}
oscillate for $\zeta^2>0$ and decay exponentially for $\zeta^2<0$. For $d=2$, the choice of constants $\gamma$, $T$, and $B$ is arbitrary \cite{zanardi17}.

For $d>2$, the analysis of the necessary and sufficient conditions for oscillating eigenvalues $\lambda(t)$ is much more complicated. Therefore, we restrict our attention to the case where $T=1/\gamma$. Now, $\zeta=2B/\sqrt{\gamma}$ is always real, and hence
\begin{equation}\label{lambdaNL}
\lambda(t)=e^{-\gamma t}\cos (B\sqrt{\gamma}t)
\end{equation}
always oscillates.

\begin{Proposition}
The generalized Pauli channel with the eigenvalues given in eq. (\ref{lambdaNL}) describes a legitimate quantum evolution if and only if
\begin{equation}\label{4}
B\leq\frac{\pi\sqrt{\gamma}}{\ln(d-1)}.
\end{equation}
\end{Proposition}

\begin{proof}
For $\lambda(t)$ in eq. (\ref{lambdaNL}), condition (\ref{Fuji-d}) reduces to
\begin{equation}
-\frac{1}{d-1}\leq e^{-\gamma t}\cos (B\sqrt{\gamma}t)\leq 1,
\end{equation}
where the second inequality trivially holds. It is enough to check that the first inequality is satisfied for the minimal value of $\lambda(t)$; namely,
\begin{equation}
-\frac{1}{d-1}\leq e^{-\gamma t_\ast}\cos (B\sqrt{\gamma}t_\ast)
\end{equation}
with $t_\ast=\pi/B\sqrt{\gamma}$. This simplifies to
\begin{equation}
\frac{1}{d-1}\geq e^{-\pi\sqrt{\gamma}/B},
\end{equation}
which is equivalent to condition (\ref{4}).
\end{proof}

Now, let us consider the Markovian evolution generated by $\mathcal{L}$ in eq. (\ref{KL}). This corresponds to $\Lambda^{\rm MS}(t)$ with $\lambda_{\alpha_\ast}^{\rm MS}(t)=1$ and $\lambda_\alpha^{\rm MS}(t)\equiv\lambda^{\rm MS}(t)$ for $\alpha\neq{\alpha_\ast}$, where
\begin{equation}
\lambda^{\rm MS}(t)=e^{-\gamma t}.
\end{equation}
The maximal fidelity $f_{\max}[\Lambda^{\rm MS}(t)]=1$, whereas the minimal fidelity
\begin{equation}
f_{\min}[\Lambda^{\rm MS}(t)]=\frac 1d \left[1+(d-1)e^{-\gamma t}\right].
\end{equation}
If minimal fidelity (\ref{minf}) oscillates, then for some $t>0$, one has
\begin{equation}
f_{\min}[\Lambda(t)]<f_{\min}[\Lambda^{\rm MS}(t)].
\end{equation}
Note that if $T=1/\gamma$, then the above inequality holds for all $t>0$. Therefore, the interesting case of increased fidelity corresponds to $T\neq 1/\gamma$. One possible choice of such parameters is shown in Fig. 1.

\FloatBarrier
\begin{figure}[ht!]
\tiny
       \includegraphics[width=0.45\textwidth]{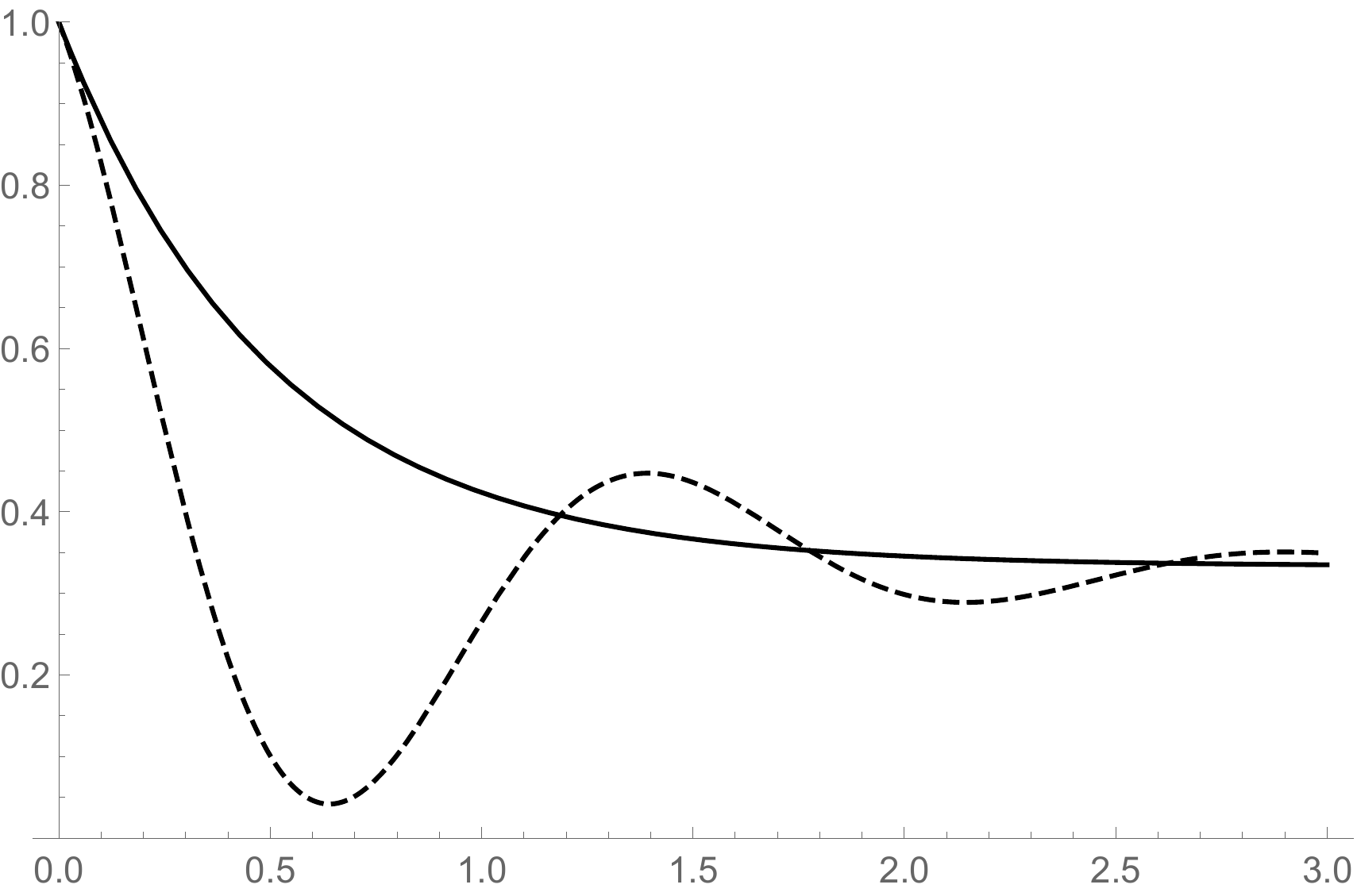}
\caption{\it The minimal channel fidelity for $d=3$, $\gamma=2s^{-1}$, and $T=2s$. The solid line corresponds to Markovian semigroup ($B=0$), and the dashed line to the general evolution with $B=3s^{-1/2}$.}
\end{figure}
\FloatBarrier

\subsection{Exponential decay I}

As the next example, let us consider the exponential functions
\begin{equation}\label{l1}
\ell_\alpha(t) = \eta e^{- \xi_\alpha t}.
\end{equation}
From Theorem \ref{TH1}, it follows that this choice leads to a legitimate dynamical map $\Lambda(t)$, provided that the coefficients $\eta$, $\xi_\alpha$ satisfy additional constraints. The necessary and sufficient conditions are presented in the following proposition.

\begin{Proposition}[\cite{kasia17}]
The functions $\ell_\alpha(t) = \eta e^{- \xi_\alpha t}$ result in a legitimate dynamical map $\Lambda(t)$ if and only if $\eta, \xi_\alpha > 0$ and
\begin{eqnarray}
 \label{1a}
  \eta \sum_{\alpha=1}^{d+1} \frac{1}{\xi_\alpha}  &\leq& \frac{d^2}{d-1} , \\ \label{2a}
   \sum_{\alpha=1}^{d+1} \frac{1}{\xi_\alpha} &\geq &  \frac{d}{\xi_\beta}.
\end{eqnarray}
\end{Proposition}

Note that the corresponding memory kernel $K(t)$ is given by
\begin{equation}\label{k1}
\begin{split}
k_\alpha(t) =& \frac 1d \eta\delta(t)
+\eta(\xi_\alpha-\eta)e^{-(\xi_\alpha-\eta)t}
\\&-\frac{1}{d}\sum_{\beta=1}^{d+1}\eta(\xi_\beta-\eta)e^{-(\xi_\beta-\eta)t}.
\end{split}
\end{equation}
%\begin{equation}
%\kappa_\alpha(t)=-\eta\left[\delta(t)+(\eta-\xi_\alpha)e^{(\eta-\xi_\alpha)t}\right],
%\end{equation}
%is unphysical when $k_\alpha(t\to\infty)\to\infty$, which occurs if there exists any $\alpha$ for which $\eta>\xi_\alpha$.
%Conditions (\ref{1a}-\ref{2a}) lead to
%\begin{equation}
%\xi_\alpha-\eta\geq-\frac{\eta}{d}.
%\end{equation}
%Hence, if $-\eta/d<\xi_\alpha-\eta<0$ for all $\alpha=1,\dots,d+1$, then the dynamical map $\Lambda(t)$ with
The associated dynamical map $\Lambda(t)$ has
\begin{equation}
\lambda_\alpha(t)=1-\frac{\eta}{\xi_\alpha}\left(1-e^{-\xi_\alpha t}\right).
\end{equation}
%is legitimate but the memory kernel $K(t)$ is not.
Now, the extreme values of the channel fidelity read
\begin{equation}\label{f1}
f_{\min}[\Lambda(t)]=1-\frac{(d-1)\eta}{d\xi_{\min}}
\left(1-e^{-\xi_{\min} t}\right),
\end{equation}
\begin{equation}\label{f2}
f_{\max}[\Lambda(t)]=1-\frac{(d-1)\eta}{d\xi_{\max}}
\left(1-e^{-\xi_{\max} t}\right).
\end{equation}
Observe that the minimal and maximal fidelities are reached at the minimal $\xi_{\min}$ and maximal $\xi_{\max}$ values of the parameters $\xi_\alpha$, respectively.

Using eq. (\ref{k1}), we decompose the memory kernel $K(t)$ into the Markovian generator
\begin{equation}\label{L1}
\mathcal{L}=\frac \eta d \sum_{\alpha=1}^{d+1}\mathcal{L}_\alpha
\end{equation}
and the memory kernel
\begin{equation}
\mathbb{K}(t)=\sum_{\alpha=1}^{d+1}\mathfrak{K}_\alpha(t)\mathcal{L}_\alpha
\end{equation}
with
\begin{equation}
\begin{split}
\mathfrak{K}_\alpha(t) =&\eta(\xi_\alpha-\eta)e^{-(\xi_\alpha-\eta)t}
\\&-\frac{1}{d}\sum_{\beta=1}^{d+1}\eta(\xi_\beta-\eta)e^{-(\xi_\beta-\eta)t}.
\end{split}
\end{equation}
Note that $\mathbb{K}(t)$ never produces legitimate solutions. Now, let us consider the evolution governed by $\mathcal{L}$ from eq. (\ref{L1}). For the corresponding dynamical map $\Lambda^{\rm MS}(t)$, it turns out that
\begin{eqnarray}\label{f3}
f_{\min}[\Lambda^{\rm MS}(t)]&= & f_{\max}[\Lambda^{\rm MS}(t)]\equiv f[\Lambda^{\rm MS}(t)] \nonumber \\
&=& \frac 1d \left(1+(d-1)e^{-\eta t}\right).
\end{eqnarray}
Interestingly, the fidelity $f[\Lambda^{\rm MS}(t)]$  can be lower than the fidelity $f[\Lambda(t)]$ for the evolution with non-local noise.

\begin{Proposition}\label{PP}
At any given $t>0$, the minimal channel fidelities in eqs. (\ref{f1}) and (\ref{f3}) satisfy the following inequalities:
\begin{enumerate}[label={(\arabic*)}]
\item if $\xi_{\min}<\eta$, then $f_{\min}[\Lambda(t)]<f[\Lambda^{\rm MS}(t)]$,

\item if $\xi_{\min}=\eta$, then $f_{\min}[\Lambda(t)]=f[\Lambda^{\rm MS}(t)]$,

\item if $\xi_{\min}>\eta$, then $f_{\min}[\Lambda(t)]>f[\Lambda^{\rm MS}(t)]$.
\end{enumerate}
%\begin{equation}
%\begin{split}
%&f_{\min}[\Lambda_{NL}(t;\xi_{\min}<\eta)]<f_{\max}[\Lambda_{NL}(t;\xi_{\max}<\eta)]\\&<
%f[\Lambda_L(t)]=f_{\min}[\Lambda_{NL}(t;\xi_{\min}=\eta)]=f_{\max}[\Lambda_{NL}(t;\xi_{\max}=\eta)]\\&<
%f_{\min}[\Lambda_{NL}(t;\xi_{\min}>\eta)]<f_{\max}[\Lambda_{NL}(t;\xi_{\max}>\eta)].
%\end{split}
%\end{equation}
\end{Proposition}

\begin{proof}
To prove the inequality $f_{\min}[\Lambda(t)]<f[\Lambda^{\rm MS}(t)]$ for $\xi_{\min}<\eta$, we show that
\begin{equation}
\begin{split}
&f_{\min}[\Lambda(t)]-f[\Lambda^{\rm MS}(t)]\\&
=\eta\frac{d-1}{d}\left[\frac{1-e^{-\eta t}}{\eta}-\frac{1-e^{-\xi_{\min}t}}{\xi_{\min}}\right]<0.
\end{split}
\end{equation}
This follows from the fact that the function
\begin{equation}
h(t;A)=\frac{1-e^{-A t}}{A}
\end{equation}
is monotonically decreasing with the increase of $A$ at a fixed $t>0$.
The proofs of the remaining relations are analogical.
\end{proof}

\begin{Proposition}\label{PP2}
At any given $t>0$, the maximal channel fidelities in eqs. (\ref{f2}) and (\ref{f3}) satisfy the following inequalities:
\begin{enumerate}[label={(\arabic*)}]
\item if $\xi_{\max}<\eta$, then $f_{\max}[\Lambda(t)]<f[\Lambda^{\rm MS}(t)]$,

\item if $\xi_{\max}=\eta$, then $f_{\max}[\Lambda(t)]=f[\Lambda^{\rm MS}(t)]$,

\item if $\xi_{\max}>\eta$, then $f_{\max}[\Lambda(t)]>f[\Lambda^{\rm MS}(t)]$,
\end{enumerate}
\end{Proposition}

%The proof is analogical to the proof of Proposition \ref{PP}.

%\FloatBarrier
%\begin{figure}[ht!]
%\tiny
%       \includegraphics[width=0.45\textwidth]{FIG_1B_.pdf}
%\caption{\it The extremal channel fidelities for $d=3$ and $\eta=1s^{-1}$. The continuous line corresponds to the local evolution, the dotted lines to the non-local evolution with $\xi_{\min}=7/9s^{-1}<\eta$ and $\xi_{\max}=35/36s^{-1}<\eta$, and the dashed lines to the non-local evolution with $\xi_{\min}=5/4s^{-1}>\eta$ and $\xi_{\max}=5/2s^{-1}>\eta$.}\label{exp}
%\end{figure}
%\FloatBarrier

\subsection{Exponential decay II}

This time, consider
\begin{equation}
\ell_\alpha(t)=\eta_\alpha\ell(t).
\end{equation}
For the exponential function $\ell(t)=e^{-\xi t}$, Theorem \ref{TH1} leads to the following proposition.

\begin{Proposition}
The functions $\ell_\alpha(t) = \eta_\alpha e^{- \xi t}$ produce a legitimate dynamical map $\Lambda(t)$ if and only if $\eta_\alpha, \xi > 0$ and
\begin{equation}\label{a3}
d\max_\beta\eta_\beta\leq\sum_{\alpha=1}^{d+1}\eta_\alpha\leq\frac{d^2\xi}{d-1}.
\end{equation}
\end{Proposition}

Now, the associated memory kernel $K(t)$ has
\begin{equation}\label{k2}
\begin{split}
k_\alpha(t) =& \frac 1d \left(\sum_{\beta=1}^{d+1}\eta_\beta-d\eta_\alpha\right)\delta(t)
+\eta_\alpha(\xi-\eta_\alpha)e^{-(\xi-\eta_\alpha)t}
\\&-\frac{1}{d}\sum_{\beta=1}^{d+1}\eta_\beta(\xi-\eta_\beta)e^{-(\xi-\eta_\beta)t},
\end{split}
\end{equation}
%is not physical if $k_\alpha(t\to\infty)\to\infty$. This is the case as long as there exists at least one $\alpha$ for which $\eta_\alpha>\xi$.
%Similarly to the previous example, if $-\eta_\alpha/d<\xi_\alpha-\eta<0$ for all $\alpha=1,\dots,d+1$, then the dynamical map $\Lambda(t)$ with the eigenvalues
whereas the eigenvalues of the dynamical map $\Lambda(t)$ are equal to
\begin{equation}
\lambda_\alpha(t)=1-\frac{\eta_\alpha}{\xi}\left(1-e^{-\xi t}\right).
\end{equation}
%is legitimate, unlike the memory kernel $K(t)$.
Finally, the minimal and maximal channel fidelities are given by
\begin{equation}\label{f4}
f_{\min}[\Lambda(t)]=1-\frac{(d-1)\eta_{\max}}{d\xi}
\left(1-e^{-\xi t}\right),
\end{equation}
\begin{equation}\label{f5}
f_{\max}[\Lambda(t)]=1-\frac{(d-1)\eta_{\min}}{d\xi}
\left(1-e^{-\xi t}\right),
\end{equation}
where $\eta_{\min}=\min_\alpha\eta_\alpha$ and $\eta_{\max}=\max_\alpha\eta_\alpha$.

The memory kernel $K(t)$ is decomposable into the Markovian generator
{\begin{equation}\label{L2}
\mathcal{L}=\frac 1d \sum_{\alpha=1}^{d+1}
\left(\sum_{\beta=1}^{d+1}\eta_\beta-d\eta_\alpha\right)\mathcal{L}_\alpha
\end{equation}}
and the memory kernel
\begin{equation}
\mathbb{K}(t)=\sum_{\alpha=1}^{d+1}\mathfrak{K}_\alpha(t)\mathcal{L}_\alpha,
\end{equation}
where
\begin{equation}
\begin{split}
\mathfrak{K}_\alpha(t) =&\eta_\alpha(\xi-\eta_\alpha)e^{-(\xi-\eta_\alpha)t}
\\&-\frac{1}{d}\sum_{\beta=1}^{d+1}\eta_\beta(\xi-\eta_\beta)e^{-(\xi-\eta_\beta)t}
\end{split}
\end{equation}
Again, the master equation with the kernel $\mathbb{K}(t)$ never produces legitimate solutions. Note that generator (\ref{L2}) has a more complicated structure than generator (\ref{L1}). It generates the dynamical map $\Lambda^{\rm MS}(t)$, for which the minimal and maximal fidelities do not coincide but are equal to
\begin{equation}\label{f6}
f_{\min}[\Lambda^{\rm MS}(t)]=\frac 1d \left(1+(d-1)e^{-\eta_{\max} t}\right),
\end{equation}
\begin{equation}\label{f7}
f_{\max}[\Lambda^{\rm MS}(t)]=\frac 1d \left(1+(d-1)e^{-\eta_{\min} t}\right).
\end{equation}
Let us analyze the above fidelities in comparison with $f_{\min}[\Lambda(t)]$ and $f_{\max}[\Lambda(t)]$ given in eqs. (\ref{f4}) and (\ref{f5}).

\begin{Proposition}
At any given $t>0$, the minimal channel fidelities in eqs. (\ref{f4}) and (\ref{f6}) satisfy the following inequalities:
\begin{enumerate}[label={(\arabic*)}]
\item if $\xi<\eta_{\max}$, then $f_{\min}[\Lambda(t)]<f_{\min}[\Lambda^{\rm MS}(t)]$,

\item if $\xi=\eta_{\max}$, then $f_{\min}[\Lambda(t)]=f_{\min}[\Lambda^{\rm MS}(t)]$,

\item if $\xi>\eta_{\max}$, then $f_{\min}[\Lambda(t)]>f_{\min}[\Lambda^{\rm MS}(t)]$,

\end{enumerate}
%\begin{equation}
%\begin{split}
%&f_{\min}[\Lambda_{NL}(t;\eta_{\max}>\xi)]<f_{\min}[\Lambda_L(t)]\\&=
%f_{\min}[\Lambda_{NL}(t;\eta_{\max}=\xi)]<f_{\min}[\Lambda_{NL}(t;\eta_{\max}<\xi)]\\&<
%f_{\max}[\Lambda_{NL}(t;\eta_{\min}>\xi)]<f_{\max}[\Lambda_L(t)]\\&=
%f_{\max}[\Lambda_{NL}(t;\eta_{\min}=\xi)]<f_{\max}[\Lambda_{NL}(t;\eta_{\min}<\xi)].
%\end{split}
%\end{equation}
\end{Proposition}

\begin{Proposition}
At any given $t>0$, the minimal channel fidelities in eqs. (\ref{f4}) and (\ref{f6}) satisfy the following inequalities:
\begin{enumerate}[label={(\arabic*)}]
\item if $\xi<\eta_{\min}$, then $f_{\max}[\Lambda(t)]<f_{\max}[\Lambda^{\rm MS}(t)]$,

\item if $\xi=\eta_{\min}$, then $f_{\max}[\Lambda(t)]=f_{\max}[\Lambda^{\rm MS}(t)]$,

\item if $\xi>\eta_{\min}$, then $f_{\max}[\Lambda(t)]> f_{\max}[\Lambda^{\rm MS}(t)]$,

\end{enumerate}
\end{Proposition}

%The proofs are analogical to the proof to Proposition \ref{PP}.

%\FloatBarrier
%\begin{figure}[ht!]
%\tiny
%       \includegraphics[width=0.45\textwidth]{FIG_2B_.pdf}
%\caption{\it The extremal channel fidelities for $d=3$, $\eta_{\min}=1s^{-1}$, and $\eta_{\max}=4/3s^{-1}$. The continuous lines correspond to the local evolution, the dashed lines to the non-local evolution with $\xi=53/54s^{-1}<\eta_{\min}$, and the dotted lines to the non-local evolution with $\xi_{\max}=3/2s^{-1}>\eta_{\max}$. The lines do not cross since $t\simeq 4s$.}\label{exp2}
%\end{figure}
%\FloatBarrier

\begin{Example}
The choice of
\begin{equation}
\xi=d,\qquad \eta_\alpha=\frac{d}{1-x_\alpha}
\end{equation}
corresponds to the convex combination of Markovian semigroups \cite{kasia17}
\begin{equation}\label{CC}
\begin{split}
\Lambda(t)=&\sum_{\alpha=1}^{d+1}x_\alpha e^{dt\mathcal{L}_\alpha}=\frac 1d \Bigg[(1+[d-1]e^{-dt})\oper\\&+(1-e^{-dt})\sum_{\alpha=1}^{d+1}x_\alpha \mathbb{U}_\alpha\Bigg],
\end{split}
\end{equation}
where $x_\alpha$ is the probability distribution. The extremal values of the associated channel fidelity are
\begin{equation}
f_{\min}[\Lambda(t)]=1-\frac{d-1}{d(1-x_{\min})}\left(1-e^{-dt}\right),
\end{equation}
\begin{equation}
f_{\max}[\Lambda(t)]=1-\frac{d-1}{d(1-x_{\max})}\left(1-e^{-dt}\right),
\end{equation}
where $x_{\min}=\min_\alpha x_\alpha$ and $x_{\max}=\max_\alpha x_\alpha$.
Note that the corresponding Markovian semigroup is governed by
\begin{equation}
\mathcal{L}=\sum_{\alpha=1}^{d+1}(1-x_\alpha)^{-1}\mathcal{L}_\alpha ,
\end{equation}
which leads to the dynamical map
\begin{equation}
\begin{split}
\Lambda^{\rm MS}(t)=&\frac{1}{d^2}
\Bigg[(1+[d-1]\sum_{\alpha=1}^{d+1}\lambda_\alpha^{\rm MS}(t))\oper
\\&+\sum_{\alpha=1}^{d+1}(1+d\lambda_\alpha^{\rm MS}(t)-\sum_{\beta=1}^{d+1}\lambda_\beta^{\rm MS}(t))
\mathbb{U}_\alpha\Bigg]
\end{split}
\end{equation}
with
\begin{equation}
\lambda_\alpha^{\rm MS}(t)=\exp\left[-\left(\sum_{\beta=1}^{d+1}\frac{1}{1-x_\beta}
-\frac{1}{1-x_\alpha}\right)t\right].
\end{equation}
Observe that
\begin{equation}
f_{\min}[\Lambda^{\rm MS}(t)]=\frac 1d \left(1+(d-1)e^{-\frac{dt}{1-x_{\min}}}\right),
\end{equation}
\begin{equation}
f_{\max}[\Lambda^{\rm MS}(t)]=\frac 1d \left(1+(d-1)e^{-\frac{dt}{1-x_{\max}}}\right).
\end{equation}
Moreover, $\eta_\alpha\geq\xi$, and therefore adding non-local effects always results in the minimal and maximal fidelities that are lower than or equal to the fidelities for the Markovian semigroup. In particular, for the eternally non-Markovian evolution ($x_\alpha=1/d$ for $\alpha=1,\dots,d$, $x_{d+1}=0$), one has $f_{\max}[\Lambda(t)]=f_{\max}[\Lambda^{\rm MS}(t)]$.
\end{Example}

\section{Engineering evolution of other quantities}

%entanglement, entropy, and coherence}}

\subsection{Quantum entanglement}

Let us examine the effects of sending one qudit of an entangled pair through the generalized Pauli channel. We analyze the evolution of entanglement $\rho_W\longmapsto\rho_W(t)=(\oper\otimes\Lambda(t))[\rho_W]$
for the maximally entangled state
\begin{equation}\label{Wer}
\rho_W=|\Phi^+\>\<\Phi^+|,\qquad|\Phi^+\>=\frac{1}{\sqrt{d}}\sum_{k=0}^{d-1}|k\>\otimes|k\>.
\end{equation}
In the case of qubits ($d=2$), one can measure entanglement using Wootters' concurrence \cite{Wootters1,Wootters2}
\begin{equation}
C(\rho)=\max\{0,\sqrt{r_1}-\sqrt{r_2}-\sqrt{r_3}-\sqrt{r_4}\},
\end{equation}
where $r_1\geq r_2\geq r_3\geq r_4$ are the eigenvalues of $\rho(\sigma_2\otimes\sigma_2)\overline{\rho}(\sigma_2\otimes\sigma_2)$.
Under the action of the Pauli channel $\Lambda(t)$, the concurrence of $\rho_W$ changes as follows,
\begin{equation}
\begin{split}
C[\rho_W(t)]=\frac 12 \max\Bigg\{0,&|\lambda_1(t)-\lambda_2(t)|-1-\lambda_3(t),\\
&|\lambda_1(t)+\lambda_2(t)|-1+\lambda_3(t)\Bigg\},
\end{split}
\end{equation}
with $\lambda_\alpha(t)$ being the eigenvalues of $\Lambda(t)$ to the Pauli matrices $\sigma_\alpha$. Observe that this formula reduces to
\begin{equation}
C[\rho_W(t)]=\frac 12 \max\left\{0,3\lambda_1(t)-1\right\}
\end{equation}
for $\lambda_3(t)=\lambda_2(t)=\lambda_1(t)\geq 0$, or to
\begin{equation}
C[\rho_W(t)]=\frac 12 |\lambda_i(t)+\lambda_j(t)|
\end{equation}
for $\lambda_k(t)=1$, $\{i,j,k\}=\{1,2,3\}$.

\begin{Example}\label{EX1}
Let us analyze how the concurrence of $\rho_W(t)$ changes depending on the type of channel $\Lambda(t)$. For the Markovian semigroup evolution with $\lambda_\alpha(t)=e^{-\eta t}$, one has
\begin{equation}\label{11}
C[\rho_W(t)]=\frac 12 \max\left\{0,3e^{-\eta t}-1\right\},
\end{equation}
which describes exponential decay until $t=\ln 3/\eta$. Now, introduce one of two different types of noise. The memory kernel evolution with exponentially decaying functions $\ell_\alpha(t)=\eta e^{-\xi t}$ leads to
\begin{equation}\label{22}
C[\rho_W(t)]=\frac 12 \max\left\{0,2-3\frac{\eta}{\xi}\left(1-e^{-\xi t}\right)\right\},
\end{equation}
where again we observe exponentional decay. However, the parameters can be chosen in such a way that the concurrence in eq. (\ref{22}) is decaying slower than in eq. (\ref{11}). In the second case, where
\begin{equation}
\ell_\alpha(t)=e^{-\eta t}\left[B\sqrt{\eta}\sin(B\sqrt{\eta}t)
+\gamma\cos(B\sqrt{\eta}t)\right]
\end{equation}
for $\alpha\neq\alpha_\ast$, the concurrence is simply
\begin{equation}
C[\rho_W(t)]=e^{-\eta t}|\cos (B\sqrt{\eta}t)|.
\end{equation}
There are two important observations to be made regarding the above formula. First, there is not a single moment in time beyond which the system is always in a separable state. Second, the state of the system becomes separable after constant periods of time $\Delta t=\pi/B\sqrt{\eta}$.
\end{Example}

The results of Example \ref{EX1} are shown in Fig. \ref{conc}. Note that adding noise to the evolution of a quantum system can prolong the entanglement or even lead to its revival.

\FloatBarrier
\begin{figure}[ht!]
\tiny
       \includegraphics[width=0.45\textwidth]{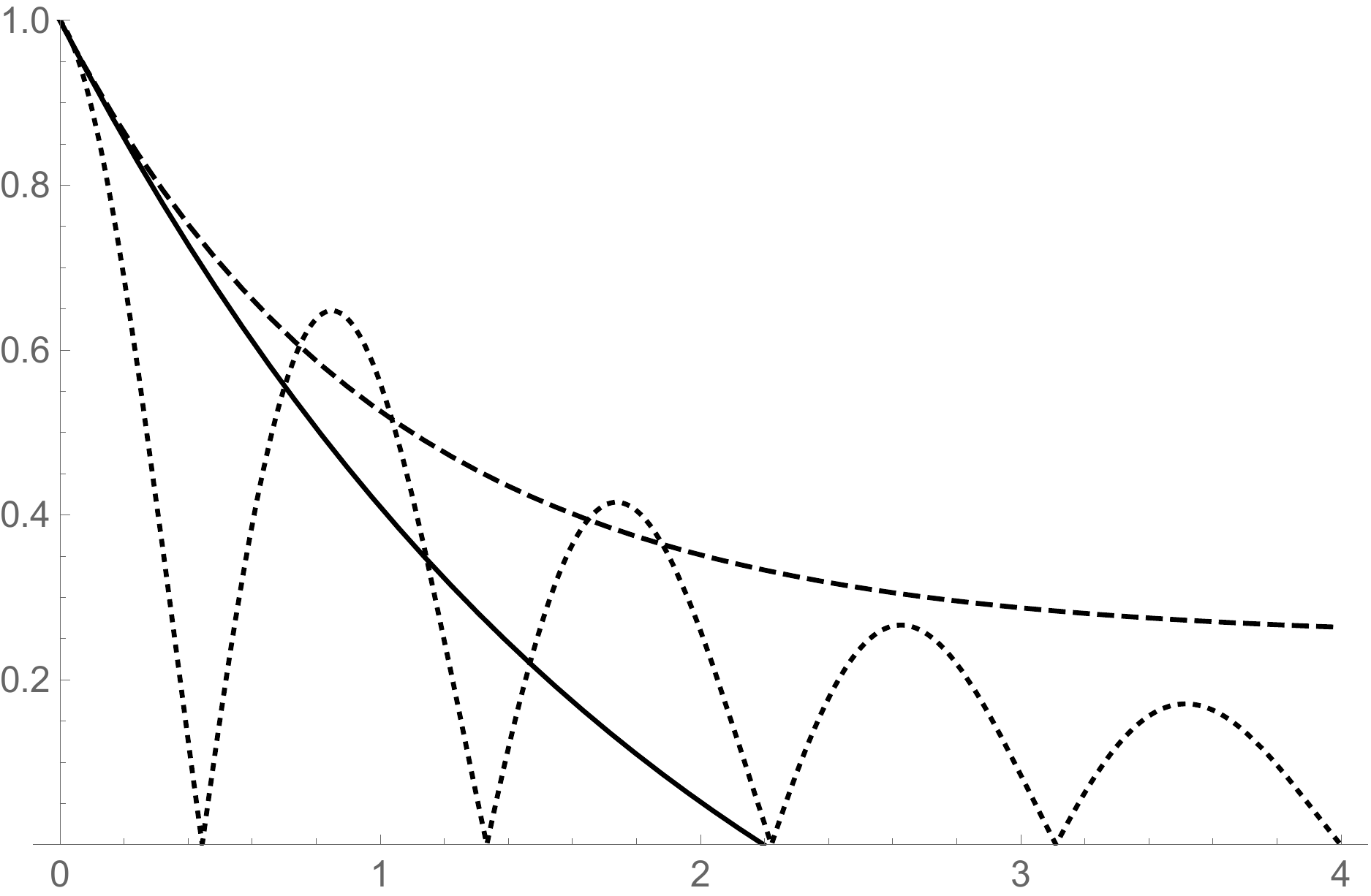}
\caption{\it The concurrence for $\eta=1/2s^{-1}$, $\xi=1s^{-1}$, and $B=5s^{-1/2}$. The continuous line corresponds to the local evolution, the dotted line to the non-local evolution with oscillations, and the dashed line to the non-local evolution with exponential decay.}\label{conc}
\end{figure}
\FloatBarrier

For $d>2$, there are unfortunately no known entanglement measures that detect all entangled states and are analytically computable. Consider the logarithmic negativity \cite{LN,MPlenio}
\begin{equation}
\mathcal{N}(\rho)=\log_2||\rho^{T_2}||_1,\qquad||\rho||_1=\Tr\sqrt{\rho^\dagger \rho},
\end{equation}
where $\rho^{T_2}$ denotes the partial transposition with respect to the second subsystem. Note that this measure does not detect PPT states. For $d=3$, the trace norm of the partially transposed $\rho_W$ is equal to
\begin{equation}\label{norm}
\begin{split}
||\rho_W^{T_2}(t)||_1=\frac 16 \Big[2|1-\lambda_0(t)|+|2+\lambda_0(t)
+\sqrt{Z(t)}|&\\+|2+\lambda_0(t)-\sqrt{Z(t)}|\Big],&
\end{split}
\end{equation}
where $\lambda_0(t)=\sum_{\alpha=1}^{4}\lambda_\alpha(t)$ and
\begin{equation}
Z(t)=9\sum_{\alpha=1}^{4}\lambda_\alpha^2(t)
-6\sum_{\alpha=1}^4\sum_{\beta>\alpha}\lambda_\alpha(t)\lambda_\beta(t).
\end{equation}
Note that eq. (\ref{norm}) simplifies to
\begin{equation}
||\rho_W^{T_2}(t)||_1= \left\{ \begin{array}{cl}
\frac 13 [1+8\lambda_1(t)] & , \ \lambda_1(t)\geq\frac 14,\\
1 & , \  \lambda_1(t)<\frac 14
\end{array}  \right.
\end{equation}
for $\lambda_1(t)=\lambda_2(t)=\lambda_3(t)=\lambda_4(t)$ and
\begin{equation}
||\rho_W^{T_2}(t)||_1=\left\{ \begin{array}{cl}
1+2\lambda(t) & , \ \lambda(t)\geq 0,\\
1 & , \  \lambda(t)<0
\end{array}  \right.
\end{equation}
for $\lambda_{\alpha_\ast}(t)=1$, $\lambda_\alpha(t)=\lambda(t)$ ($\alpha\neq\alpha_\ast$).

\begin{Example}\label{EX2}
Let us analyze the behaviour of the trace norm $||\rho_W^{T_2}(t)||_1$ under the types of evolution considered in Example \ref{EX1}. For the Markovian semigroup evolution,
\begin{equation}\label{33}
||\rho_W^{T_2}(t)||_1=
\left\{ \begin{array}{cl}
1 & , \ e^{-\eta t}\leq\frac 14,\\
\frac 13 \left(1+8e^{-\eta t}\right) & , \ e^{-\eta t}>\frac 14
\end{array}  \right.
\end{equation}
decays exponentially until $t=\ln 4/\eta$. If the functions $\ell_\alpha(t)$ decay exponentially, then
\begin{equation}
||\rho_W^{T_2}(t)||_1=
\left\{ \begin{array}{cl}
1 & , \ e^{-\xi t}\leq 1-\frac{3\xi}{4\eta},\\
3-\frac{8\eta}{3\xi}\left(1-e^{-\xi t}\right) & , \ e^{-\xi t}> 1-\frac{3\xi}{4\eta}.
\end{array}  \right.
\end{equation}
Finally, for oscillating $\ell_\alpha(t)$ with $\gamma=1/T\equiv\eta$,
\begin{equation}\label{44}
||\rho_W^{T_2}(t)||_1=
\left\{ \begin{array}{cl}
1 & , \ \cos(B\sqrt{\eta}t)\leq 0,\\
1+2e^{-\eta t}\cos(B\sqrt{\eta}t) & , \ \cos(B\sqrt{\eta}t)>0.
\end{array}  \right.
\end{equation}
\end{Example}

The logarithmic negativity for Example \ref{EX2} is plotted in Fig. \ref{ln}. Comparing with the concurrence for $d=2$, the logarithmic negativity for oscillating functions vanishes for longer moments. Therefore, the state of the system remains either separable or PPT for finite periods of time.

\FloatBarrier
\begin{figure}[ht!]
\tiny
       \includegraphics[width=0.45\textwidth]{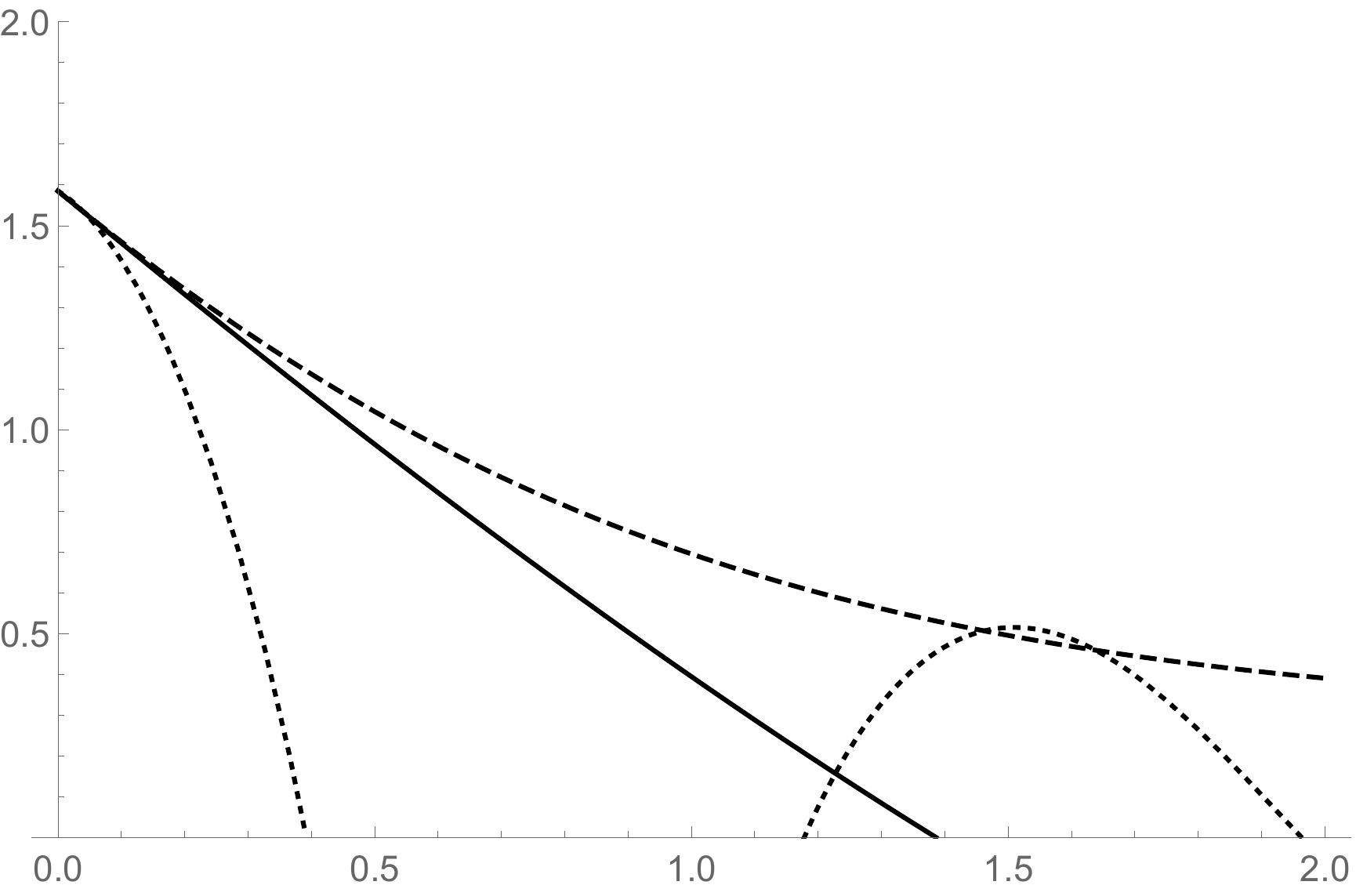}
\caption{\it The logarithmic negativity for $\eta=1s^{-1}$, $\xi=3/2s^{-1}$, and $B=4s^{-1/2}$. The continuous line corresponds to the local evolution, the dotted line to the non-local evolution with oscillations, and the dashed line to the non-local evolution with exponential decay.}\label{ln}
\end{figure}
\FloatBarrier

\subsection{Entropy}

Now, let us analyze the evolution of entropy for the projector $P_k^{(\alpha)}$ onto the mutually unbiased basis vector under the generalized Pauli channels $\Lambda(t)$. Observe that $\Lambda(t)$ transforms $P_k^{(\alpha)}$ into
\begin{equation}\label{rhoE}
\rho_{k,\alpha}(t)=\Lambda(t)[P_k^{(\alpha)}]=\sum_{j=0}^{d-1}\nu_jP_j^{(\alpha)},
\end{equation}
where
\begin{align}
&\nu_k=\frac 1d [1+(d-1)\lambda_\alpha(t)],\\
&\nu_j=\frac 1d [1-\lambda_\alpha(t)],\qquad j\neq k.
\end{align}
Therefore, the von Neumann entropy of the output state $\rho_{k,\alpha}(t)$ reads
\begin{equation}
S[\rho_{k,\alpha}(t)]=-\nu_k\ln\nu_k-(d-1)\nu_j\ln\nu_j,\qquad j\neq k.
\end{equation}
Note that its value depends only on the eigenvalue $\lambda_\alpha(t)$ for the distinguished $\alpha$. In Fig. \ref{entr}, it is shown that adding noise to the time-local evolution can bring more order to the system, which manifests itself in lower entropy.

\FloatBarrier
\begin{figure}[ht!]
\tiny
       \includegraphics[width=0.45\textwidth]{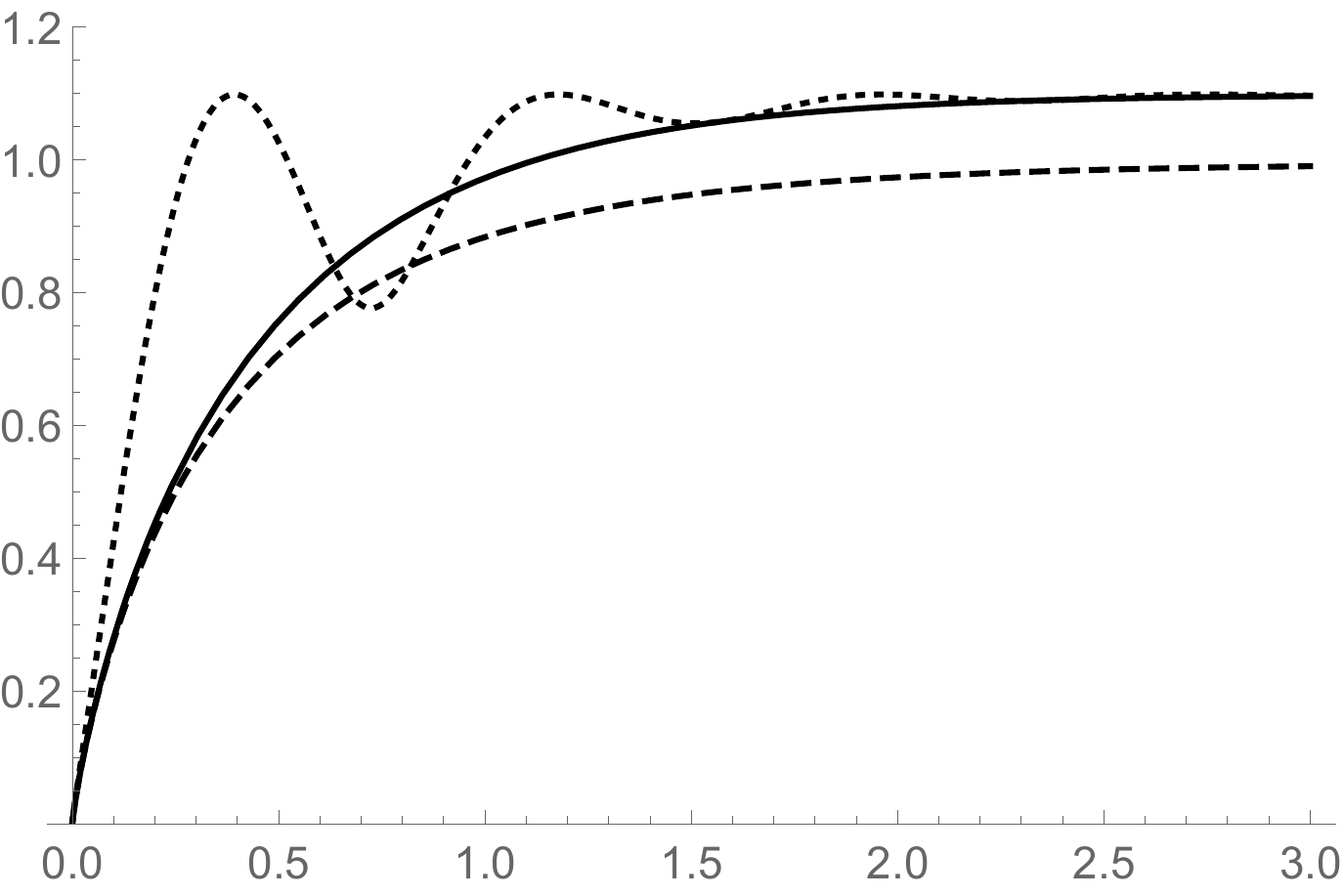}
\caption{\it The von Neumann entropy for $d=3$, $\eta=1s^{-1}$, $\xi=3/2s^{-1}$, $T=s$, and $B=4s^{-1/2}$. The continuous line corresponds to the local evolution, the dotted line to the non-local evolution with oscillations, and the dashed line to the non-local evolution with exponential decay.}\label{entr}
\end{figure}
\FloatBarrier

\subsection{Quantum coherence}

Consider now the evolution of quantum coherence. The most popular measure of quantum coherence is provided by $l_1$-norm  \cite{Cramer,Plenio}
\begin{equation}
\mathcal{C}_{l_1}[\rho]=\sum_{i=0}^{d-1}\sum_{j\neq i}|\<i|\rho|j\>|.
\end{equation}
For $d=3$ and $\rho_{k,\alpha}(t)$ given in eq. (\ref{rhoE}), one easily finds
\begin{equation}
\mathcal{C}_{l_1}[\rho_{k,\alpha}(t)]=
\left\{ \begin{array}{cl}
0 & , \ P_k^{(\alpha)}=|k\>\<k|,\\
\lambda_\alpha(t) & , \ P_k^{(\alpha)}\neq|k\>\<k| ,
\end{array}  \right.
\end{equation}
and it is plotted in Fig. \ref{coh}.

%Unfortunately, this result is not generalizable for $d>3$ in a straightforward manner.

\begin{figure}[ht!]  
\tiny
       \includegraphics[width=0.45\textwidth]{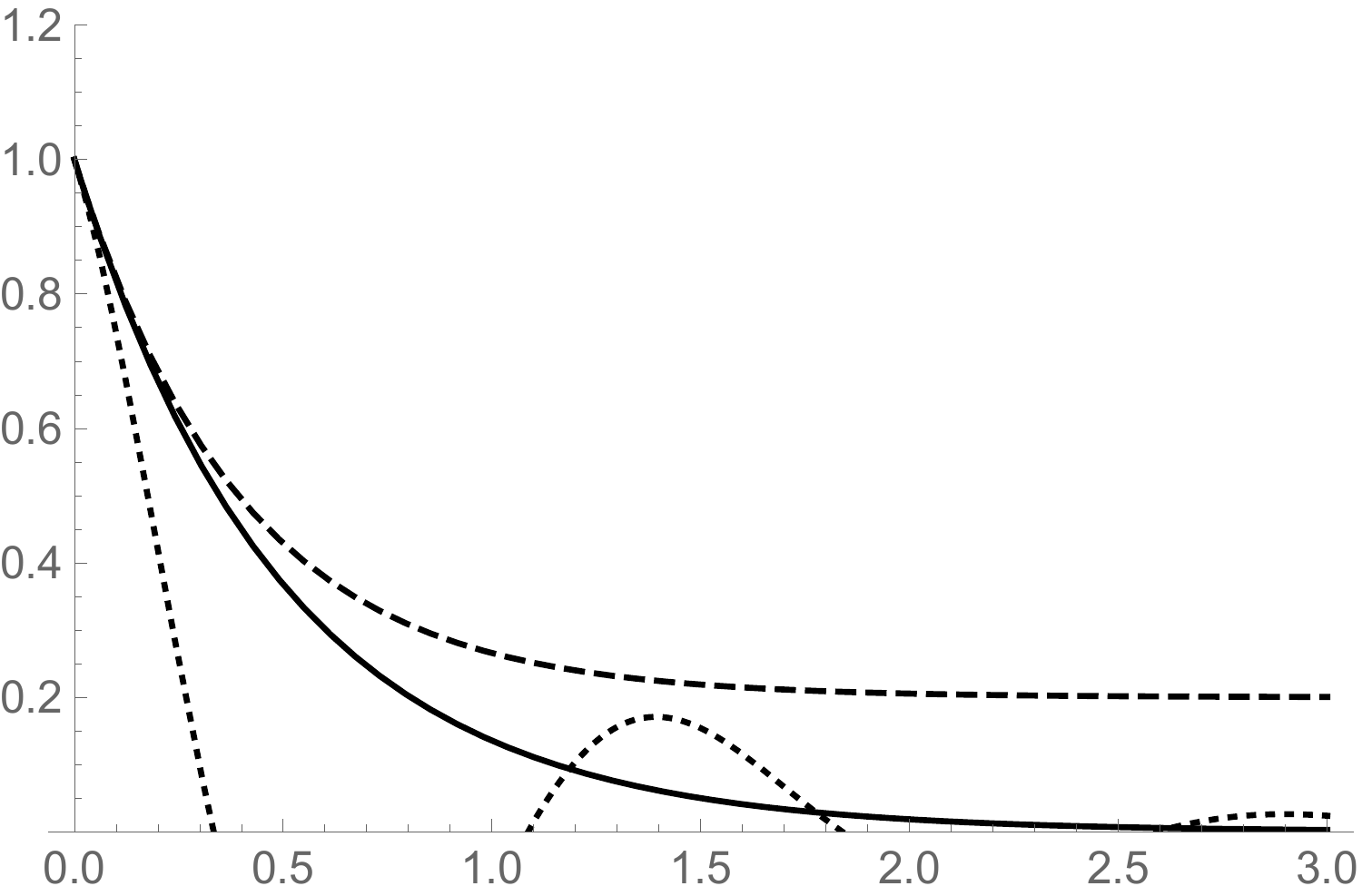}
\caption{\it The $l_1$-norm of coherence for $\eta=2s^{-1}$, $\xi=5/2s^{-1}$, $T=2s$, and $B=3s^{-1/2}$. The continuous line corresponds to the local evolution, the dotted line to the non-local evolution with oscillations, and the dashed line to the non-local evolution with exponential decay.}\label{coh}
\end{figure}
\FloatBarrier

\section{Conclusions}

We analyzed the channel fidelity of the generalized Pauli channels, which measures the distortion between the pure input and output quantum states.
We compared the evolution of fidelity for the Markovian semigroup $\Lambda^{\rm SM}$ generated by the GKSL generator $\mathcal{L}$with the general dynamical map $\Lambda(t)$ generated by the non-local memory kernel master equation with $K(t) = \delta(t)\mathcal{L} + \mathbb{K}(t)$ (with the same local part $\mathcal{L}$). It turns out that introducing non-local environmental noise $\mathbb{K}(t)$ to the Markovian evolution can increase the fidelity of the time-dependent channel $\Lambda(t)$. In other words, this results in the output states that are less distorted. Also, the additional noise can help to preserve entanglement for longer periods of time, as well as decrease the entropy and increase  the coherence of quantum states. Therefore, we showed that sending quantum information through the generalized Pauli channel generated by a non-local memory kernel can be more effective than through the channel generated by a purely Markovian generator. 
%This conclusion indicates that there is an alternative to fighting the environmental noise. Instead, one can benefit from the non-Markovian memory effects by preserving the quantum information for a longer period of time. 
These results support many other observations that a proper engineering of noise can be beneficial  for quantum information processing.

% where it has already been noticed that non-Markovian processes can improve the fidelity of quantum channels.

It would be interesting to investigate how the non-local noise $\mathbb{K}(t)$ can influence not only the channel fidelity but also the channel capacity. The latter problem is much more difficult due to the very nontrivial definition of the channel capacity. Some results in this direction were already derived in \cite{Bogna1}, where it was shown that non-Markovian memory effects can increase quantum capacity. It would be interesting to study the capacity problem for the generalized Pauli channels as well. Another issue is related to the very structure of the noise operator $\mathbb{K}(t)$. In the examples presented in this paper the memory kernel  $K(t) = \delta(t)\mathcal{L} + \mathbb{K}(t)$ generates legitimate quantum evolution, however the noise kernel $\mathbb{K}(t)$ alone does not. It would be interesting to search for the purely non-local noise kernels $\mathbb{K}(t)$ which generate legitimate dynamical maps. 

%There are still many open questions regarding the channel fidelity. 
%Observe that in Sections 4.2 and 4.3, the memory kernel master equations with the memory kernels $\mathbb{K}(t)$ lead, in general, to unphysical results. It would be interesting to find more examples where $\mathbb{K}(t)$ are legitimate memory kernels. Also, one could see what happens if the environmental noise is realized by introducing not the memory kernel $K(t)$ but the time-local generator $\mathcal{L}(t)$.

\section*{Acknowledgements} Authors were supported by the National Science Centre project 2015/17/B/ST2/02026.

%\bibliographystyle{mojformat2}
%\bibliography{bibliography}

%
%\section*{Appendix}
%
%\begin{Proposition}\label{prop}
%For $d>2$, the functions $\ell_\alpha(t)$ in eq. (\ref{ell}) satisfy the conditions in Theorem \ref{TH1} if
%\begin{equation}\label{pp}
%B\leq\frac{1}{\sqrt{d(d-2)T}}.
%\end{equation}
%\end{Proposition}
%
%\begin{proof}
%The eigenvalues $\lambda(t)$ describe a legitimate dynamical map if and only if the Fujiwara-Algoet conditions in (\ref{Fuji-d}) hold. Here, they are equivalent to
%\begin{equation}\label{FA}
%-\frac{1}{d-1}\leq\lambda(t)\leq 1.
%\end{equation}
%By differentiating $\lambda(t)$, we find that the local extrema are reached at $t_\ast\geq 0$ such that
%\begin{equation}
%\cos\left(\frac{\zeta t_\ast}{2T}+\arctan\frac{\gamma T-1}{\zeta}\right)=\pm
%\frac{\zeta}{\sqrt{\zeta^2+(1+\gamma T)^2}}.
%\end{equation}
%Now, observe that the extremal values read
%\begin{equation}
%\lambda(t_\ast)=\pm e^{-\frac{(1+\gamma T)t_\ast}{2T}}\frac{B\sqrt{T}}{\sqrt{1+B^2T}}.
%\end{equation}
%From (\ref{FA}), we see that it is enough if
%\begin{equation}
%\frac{B\sqrt{T}}{\sqrt{1+B^2T}}\leq\frac{1}{d-1}
%\end{equation}
%holds, which is equivalent to ineq. (\ref{pp}).
%\end{proof}
%
%The minimal channel fidelity oscillates only for $\zeta^2>0$; that is,
%\begin{equation}\label{osc}
%B>\frac{|1-\gamma T|}{2\sqrt{\gamma}T}.
%\end{equation}
%Therefore, {we know that the channel fidelity oscillates if $B$ belongs to the following range,}
%\begin{equation}\label{1}
%\frac{|1-\gamma T|}{2\sqrt{\gamma}T}<B\leq\frac{1}{\sqrt{d(d-2)T}}.
%\end{equation}

\end{document}